\documentclass[reqno,11pt]{amsart}
\usepackage[mathscr]{euscript} 
\usepackage{bm}

\usepackage{amssymb}
\usepackage{amsmath}
\usepackage{amsthm}
\usepackage{pgf}
\usepackage{color}
\usepackage[colorlinks=true, linkcolor=blue_links,
urlcolor=blue_links, citecolor=blue_links]{hyperref}
\definecolor{blue_links}{RGB}{13,0,180} 
\definecolor{lightblue}{rgb}{0.9,0.9,1}
\usepackage{graphicx}

\newtheorem{theorem}{Theorem}[section]
\newtheorem{lemma}[theorem]{Lemma}
\newtheorem{proposition}[theorem]{Proposition}

\theoremstyle{definition}
\begingroup

\endgroup

\newcommand{\Rz}{\mathbb{R}}
\newcommand{\Nz}{\mathbb{N}}
\newcommand{\Zz}{\mathbb{Z}}
\newcommand{\Tz}{\mathbb{T}}

\newcommand{\osc}{\text{\rm osc}\,}
\newcommand{\Lip}{\text{\rm Lip}\,}
\renewcommand{\d}{\text{\rm d}}
\newcommand{\mo}{\, {\rm mod \,1} }

\begin{document}

\title[Crystallization in a one-dimensional
  periodic landscape]{Crystallization in a one-dimensional
  \\ periodic landscape}

\author{Manuel Friedrich}
\address[Manuel Friedrich]{Applied Mathematics,  
Universit\"{a}t M\"{u}nster, Einsteinstr. 62, D-48149 M\"{u}nster, Germany}
\email{manuel.friedrich@uni-muenster.de}
\urladdr{https://www.uni-muenster.de/AMM/Friedrich/index.shtml}

\author{Ulisse Stefanelli}
\address[Ulisse Stefanelli]{Faculty of Mathematics, University of Vienna, 
Oskar-Morgenstern-Platz 1, 1090 Wien, Austria  and  Istituto di Matematica
Applicata e Tecnologie Informatiche \textit{{E. Magenes}}, v. Ferrata 1, 27100
Pavia, Italy.}
\email{ulisse.stefanelli@univie.ac.at}
\urladdr{http://www.mat.univie.ac.at/$\sim$stefanelli}


\begin{abstract} 
We consider the crystallization problem for a finite one-dimensional collection of identical
hard spheres in a periodic energy landscape. This issue
arises in connection with the investigation of crystalline states of
ionic dimers, as
well as in epitaxial growth on a crystalline substrate in presence of
lattice mismatch. 
Depending on
the commensurability of the radius of the sphere and the period of the
landscape, we  discuss the possible emergence of crystallized
states.  In particular, we prove that crystallization in
arbitrarily long chains is {\it generically} not to be expected. 

\end{abstract}

\maketitle

\pagestyle{myheadings}

\section{Introduction}
The emergence of crystalline states at low temperatures is a
common phenomenon in material systems. Its rigorous mathematical
description
poses severe mathematical challenges even at the quite simplified
setting of Molecular Mechanics, where configurations of particles 
interacting via classical potentials are considered
\cite{Friesecke-Theil15,Lewars}. Here,
{\it crystallization} corresponds to the periodicity of ground state
configurations, an instance which in many cases is still eluding a
complete mathematical understanding. In fact,
rigorous mathematical crystallization results are  scarce  and often
limited to very specific choices of data \cite{Blanc}. 
Specifically, interactions with the environment are
usually neglected or
assumed to be homogeneous. 

We intend to progress in this quest by addressing here the case of
 finite one-dimensional crystallization in
a periodic, possibly nonconstant  energy landscape. Given the configuration
$\{x_1,\dots,x_n\} \subset \Rz^n$, we consider the {\it energy}
$$E=\sum_{i}v_1(x_i) +\frac12\sum_{i\not = j}v_2(|x_i-x_j|).$$
The {\it landscape potential} $v_1: \Rz \to [0,\infty)$ is assumed to be $1$-periodic,  piecewise continuous,  and lower semicontinuous with   $\min v_1=0$. 
The {\it interaction potential} $v_2: \Rz_+\to \Rz\cup
\{\infty\} $ is of {\it hard-sphere}
type at distance $\alpha>0$, namely $v_2 =\infty$ on $[0,\alpha)$,
$ v_2(\alpha)=-1$, and $v_2 =0 $ on $(\alpha,\infty)$,  see  \cite{Heitman-Radin80}. 
A collection of $n$ particles is called an {\it $n$-crystal} (or, simply,
{\it crystal}) if it is of the
form $\{x, x+\alpha, \ldots,   x+ \alpha(n-1)\}$ for some $x\in
\Rz$, see Figure~\ref{figure}.   If all ground-state $n$-particle
configurations are  $n$-crystals,  we say that 
{\it $n$-crystallization} holds. We call {\it crystallization} the
case when
$n$-crystallization holds for all $n$. 

\begin{figure}
  \centering
  \pgfdeclareimage[width=139mm]{figure}{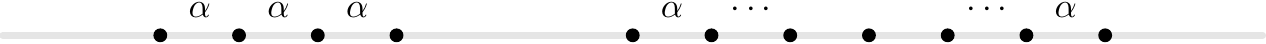}
    \pgfuseimage{figure}
\caption{A $4$-crystal and a $7$-crystal.}
 \label{figure}
\end{figure}

 The aim of this note is to investigate  crystallization
under different choices for $v_1$
and $\alpha$. Our main result states that crystallization does not {\it
  generically} hold. More precisely, we have the following.

\begin{theorem}[Generic noncrystallization]\label{thm: negative0}
For all given $\alpha$ and $v_1$ as above and each $\varepsilon>0$, there exist
$\alpha^\varepsilon$ and $v_1^\varepsilon$ as above with $|\alpha -
\alpha^\varepsilon|<\varepsilon$ and  $\| v_1 -
v_1^\varepsilon\|_{L^\infty(0,1)}<\varepsilon$,  and a strictly increasing sequence
$(n_k)_{k\in \Nz} \subset \Nz$ such that $n_k$-crystallization does
not hold for the energy $E^\varepsilon$ defined from
$\alpha^\varepsilon$ and $v_1^\varepsilon$.
\end{theorem}

In addition to this generic negative result, which is proved in
 Section  \ref{sec:gg}, we discuss different {\it nongeneric} settings
where crystallization does hold.  
Two quite different scenarios arise, depending on the rationality of
$\alpha$. 

In case $\alpha$ is rational  (a nongeneric property),
the crystallization problem can be solved by localized
arguments.  In
particular, Theorem \ref{thm:rat} states that  $n$-crystallization
holds  under some specific conditions on $v_1$ which are independent of $n$ but only  depend  on the irreducible fraction of $\alpha$.   In fact, we are able to
present a hierarchy of
sufficient conditions entailing crystallization, see Proposition
\ref{thm:hi}.

The case of $\alpha$ irrational is tackled in 
Section \ref{sec:irrational} instead. Here, the ergodic character of
the map $x \in [0,1)\mapsto (x+\alpha)_{\rm mod \,1}$ comes into
play. We resort in using and extending some tools
from the theory of {\it low discrepancy sequences} \cite{Drmota}, carefully
quantifying the extent at which the potential landscape $v_1$ is
explored by the latter map. Such quantitative information is instrumental in investigating
 crystallization.  Here, we are able to find
a specific class of landscape potentials $v_1$ entailing
crystallization, see Theorem \ref{thm: positive}  and the
discussion thereafter.  
 
 The specific form of  the  energy $E$ is inspired by the modelization
of a {\it dimer} of elements $A$ and $B$ at zero  temperature. By labelling the corresponding atoms as  $  x_i
$ and $  y_\ell $, a possible choice for the energy of the dimer is
$$
\frac12\sum_{i\not = j}v_2^A(|x_i-x_j|) + \frac12 \sum_{\ell\not =
  k} v_2^B(|y_\ell - y_k|)+ \sum_{i,\ell}v_2^{\rm
  int}(|x_i - y_\ell|). 
$$
Here, $v_2^A$ and $v_2^B$ are the intraspecific two-body interaction
energies for atoms of type $A$ and $B$,  minimized  at  the
interaction distance $0<
\alpha\not =1$ and $1$, respectively,  and  $v_2^{\rm
  int}$ is an interaction energy between types. Assume now that  type $B$
has already formed a  one-dimensional infinite rigid crystal, say
$\Zz$, see  \cite{Fanzon,Levi}  for a similar approach.  By
removing the self-interacting $v_2^B$ terms, the energy can hence be
rewritten   as a function  of
$\{x_1,\dots,x_n\}$ in the form of $E$ by letting $v_2=v_2^A$ and 
$$v_1(x) := \sum_{\ell \in \Zz}v_2^{\rm int}(|x - \ell|).$$
 By assuming  that the latter series converges for all $x \in [0,1)$,  the resulting landscape potential $v_1$ is $1$-periodic. 

Energies of the type of $E$ may also arise in modeling the {\it epitaxial growth}
of a first layer of type $A$ on top of an underlying rigid crystal of type $B$ in
presence of lattice mismatch. Here, the potential $v_1$ represents the
effect of the rigid substrate, with periodicity $1$. The deposited
layer $\{x_1,\dots,x_n\}$ is then expected to optimize intraspecific atomic interactions
in a given nontrivial potential landscape. 

Crystallization problems 
 have received constant attention in the last
decades. The reader is referred to the recent survey by {\sc Blanc \&
  Lewin} \cite{Blanc} for a comprehensive account on the literature. 
To the best of our knowledge,
crystallization results in periodic landscapes are still currently
unavailable. We contribute here in extending the classical one-dimensional
crystallization theory \cite{Gardner,Hamrick,Ventevogel} toward the
discussion of molecular compounds. 

Numerical studies  on 
crystallization in multicomponent systems are abundant,
see \cite{Assoud,Assoud2,Eldridge,Levi,Xu}, just to mention a few. On
the other hand, rigorous crystallization results for  such  systems are  scarce.  A first result in this direction is due to {\sc Radin}
\cite{Radin86}, who studies a specific multicomponent two-dimensional system showing
quasiperiodic ground states. 
{\sc B\'etermin, Kn\"upfer, \& Nolte}
\cite{Betermin} investigate conditions for crystallization of
{\it alternating} one-dimensional 
configurations interacting via a smooth interaction density $v_2$.  Two
dimensional dimer crystallization results in
hexagonal and square geometries   for a hard-spheres
interaction $v_2$ are  given  in \cite{kreutz, kreutz2}.

\section{Preliminaries}\label{sec:preliminaries}

In this section we collect  some preliminary discussion and fix notation.
 
 To start with,  one  can assume with no loss of generality that $\alpha <1$. Indeed,
if $\alpha =1$, then ground-state configurations are obviously
$n$-crystals with all particles sitting  at $x_1+ \Nz$, where $x_1
\in [0,1)$ is a minimizer of $v_1$.  Since $\min v_1 =0$, the corresponding energy is $E=(n-1)v_2(\alpha) = -(n-1)$. On the other hand,
if $\alpha >1$, we can rescale the problem by   redefining $\alpha$ as
$\alpha/\lceil \alpha \rceil\leq 1$, where $\lceil \alpha
\rceil=\min\{z \in \Zz \, : \, \alpha \leq z\}$, and by  replacing $t
\mapsto v_1(t)$ with $t \mapsto  v_1(\lceil \alpha \rceil\,t) $.  We also use the notation $(x)_{\rm mod \,1}:= x - \lfloor x \rfloor $ for all $x \in \Rz$, where $\lfloor x
\rfloor=\max\{z \in \Zz \, : \, x \geq z\}$.  Given any $A\subset \Rz$, we indicate by
$\chi_A$ the corresponding {\it characteristic function}, namely,
$\chi_A(x)=1$ if $x\in A$ and $\chi_A(x)=0$ elsewhere.

The  total contribution of the landscape  potential  to the
energy  of  the  $n$-crystal with the leftmost particle sitting at $x$  (i.e., the collection of points $\{x, x+\alpha, \ldots,   x+ \alpha(n-1)\}$)  reads  
$$ V_n(x):=\sum_{j=0}^{n-1} v_1(x+j\alpha).$$
Let us introduce the notation $V^*_n := \min V_n $ and indicate
with $x^*_n \in  [0,1) $ (possibly not  uniquely)  a minimizer, i.e.,  $V^*_n=V_n(x^*_n)$.  Note that a minimizer exists since $v_1$ is lower semicontinuous.  
If  an   $n$-particle  ground state is an
$n$-crystal, then necessarily its leftmost
particle sits at a point $x_n^*$ (possibly not unique).
Note that one always has that 
\begin{equation}\label{eq:basic}
V^*_p + V^*_q  \leq V^*_{p+q}  \quad {\forall \, p,  q   \in \Nz } 
\end{equation}
as the minimization on the right-hand side is performed under an extra constraint with
respect to those on the left-hand side. 

Our first aim is to  elucidate  the role of the somewhat {\it opposite} relation
\begin{equation}\label{suff}V^*_{p+q} < V^*_p + V^*_q +1\quad \forall \,
  p, q \ \ \text{with} \ \ 
  p+q\leq n.
\end{equation}
Under condition \eqref{suff}, one has that the splitting of an
$n$-crystal into smaller crystals is energetically not  favored. 
(In the following, we use the term {\it splitting} to refer to a
configuration made of different crystals.)   Consider indeed an $n$-particles configuration made of $j$ different
crystals $\{x^j_1,\dots,x^j_{n_j}\}$ with $n_1 + \dots +n_j= n$. 
In
case $j \geq 2$, one can use \eqref{suff}  and $\min v_2 =-1$   in order to get that 
$$  E \geq \sum\nolimits_j (V^*_{n_j} - (n_j-1)) \stackrel{\eqref{suff}}{>}
V^*_n - (j-1)- n + j =   V^*_n  -(n -1).$$
The above right-hand side is
the energy of the $n$-crystal, which is then favorable with respect to
any of its splittings, regardless of the number  $j \ge 2$  of splitting
parts.

On the other hand, condition \eqref{suff} is {\it almost} necessary
for $n$-crystallization to hold. Indeed, if one has 
\begin{equation}
V^*_{p+q} > V^*_p + V^*_q +1 \ \text{for some $p, q $ with
  $p+q =  n$,}\label{nec2}
\end{equation}
then $n$-crystals are not ground states as splitting an $n$-crystal
into a $p$- and a $q$-crystal lowers the energy.  In case 
equality holds in  condition  \eqref{nec2} for
some $p+q=n$, $n$-crystals and the union of a $p$-
and a $q$-crystal are equienergetic. In conclusion, we have checked the following.

\begin{proposition}[Key condition]\label{thm:key}
Condition \eqref{suff} implies $n$-crystallization. On the other hand,
 $n'$-crystallization for all $n'\leq n$ implies \eqref{suff}.
\end{proposition}

Owing to the latter, in order to check for the validity of $n$-crystallization, 
one is left with  checking the key condition
\eqref{suff}. This check is at the core of all our arguments in the remainder of the
paper.  It will  be investigated under different settings for $\alpha$ and $v_1$. As
already mentioned  in the Introduction, the analysis 
depends strongly on  $\alpha$ being rational or
not. Correspondingly, our discussion is  divided in the
coming Sections \ref{sec:rational} ($\alpha$ rational) and \ref{sec:irrational} ($\alpha$ irrational).

Let us conclude this section by some remarks:

 a)  At first, we would like
to record that, differently from the trivial
case $v_1  \equiv  0$, $n$-crystallization may indeed depend on the number $n$ of
particles involved. We present here an example illustrating this
fact.  
Let $\alpha=1/2$ and $v_1$ be such that 
$$v_1(x)=h({1 - 4(x-1/2)^2})\quad \ \ { \text{ for } x \in [0,1) }$$
for some $h>0$. Then, one has that 
 $$V^*_{n}= h\lfloor n/2\rfloor \quad \forall \, n \in \Nz.$$ 
In particular, given $p+q= n$,  there holds 
$$ h  (k+j)+h\ell = V^*_n , \quad  h  (k+j) = V^*_p + V^*_q,$$
where $k=\lfloor p/2 \rfloor$, $j=\lfloor q/2 \rfloor$, and $\ell =
\lfloor (n-2k-2j)/2\rfloor$. Note that  $\ell=0$ unless both $p$
and $q$ are odd, in which case $\ell=1$. Then, we conclude that the
sufficient condition \eqref{suff} holds if $h<1$ whereas, if $h >1$,
 then  condition \eqref{nec2} implies that ground states
are not crystalline for $n \geq 6$ and even. 
In case $h=1$  and $n$ even,   crystallized states and noncrystallized states are
energetically equivalent. On the other hand, ground states
are crystalline for all $n \geq 3$ odd, regardless of the value of $h$.

 b)  One may wonder if the key condition \eqref{suff} could be weakened  to  
\begin{equation}\label{suffweak}V^*_{p+q} < V^*_p + V^*_q +1\quad \forall \, 
  p,   q \ \ \text{with} \ \ 
  p+q = n,
\end{equation}
namely, by restricting to the splitting into exactly two subcrystals
at level $n$ only.
This is however not the case, as the following simple example shows.  Let $\alpha = 1/4$,  take 
$v_1=0$ on $A=\{0,1/4, 3/4  \}$, and $v_1=h>2$ out of a very small
neighborhood of $A$. One can readily compute that $V^*_1=V^*_2=0,
V^*_3=V^*_4=V^*_5=h$. Hence, condition \eqref{suffweak}
holds for $n=5$. In particular, no splitting of a $5$-crystal into
exactly two smaller crystals is favorable. On the other hand, the ground states for $n=5$ are
not crystalline as one can favorably split a $5$-crystal into two
$2$-crystals and a $1$-crystal since   $V^*_1+2 V^*_2 +2 = 2 <h=
V^*_5$. 

 c)  Eventually, we  present an example showing that, in
 general, the minimizers $x^*_p$ of $V_p$ may depend on $p$. Let $\alpha=1/4$,  $v_1 = 0$ only on $A = \lbrace 0, 1/4 \rbrace$, and $v_1 = 1$ out of a very small neighborhood of $A$.  One readily gets that $x_1^*\in\{0,1/4\}$, $x_2^*=0$,
 $x^*_3\in\{0,3/4\}$, $x_4^* \in\{0,1/4,1/2,3/4\}$,  $x^*_5 \in \lbrace 0, 1/4 \rbrace$, and  $x^*_6=0$. 
In
 particular, the position of the leftmost particle   of a crystal ground state
 may depend on its length.

\section{Crystallization for rational $\alpha$}\label{sec:rational}
Let $\alpha$ be rational. By possibly rescaling, as explained at the beginning of Section \ref{sec:preliminaries}, one can assume with no loss of generality that $m\alpha =1$ for some
$m\in \Nz$. Note that in this case 
\begin{equation}
  \label{eq:exact}
  V_{pm}^* = pV^*_m\quad \forall \, p \in \Nz. 
\end{equation}
The aim of this section is that of showing that  crystallization,
namely $n$-crystallization
{\it for all} $n $,  can be achieved under some  version of the key
condition \eqref{suff} which is {\it localized} with respect to
$n$.   More precisely, we consider  the  condition 
\begin{equation}\label{suff3}
(p+q)\frac{V^*_{m}}{m} < V^*_p + V^*_q  + 1\quad \forall \, p,  \,q \in \Nz\  \
\text{with}  \ \ 1 \le p,q \le m-1.
\end{equation}
The condition is in the same spirit as the key condition \eqref{suff},
with the difference that on the left-hand side the minimal
contribution of the landscape potential to the energy of an
$(p+q)$-crystal is replaced by   the 
 {\it minimal per-particle energy} 
$V_m^* / m$ in an $m$-crystal times the number of particles $p+q$. On
the one hand,  condition  \eqref{suff3} is  stronger 
than \eqref{suff}  since 
$$(p+q) V_m^* \stackrel{\eqref{eq:exact}}{=} V_{(p+q)m}^*\stackrel{\eqref{eq:basic}}{\ge} m V_{p+q}^*,    $$
and therefore $(p+q) V_m^* /m \ge V^*_{p+q}$. On the other hand,   in comparison with \eqref{suff}, condition \eqref{suff3}
   involves only the energy of $p$-crystals of length at most $m$,  and is
thus weaker and  easier to
check.  In particular, it is {\it local} with respect to $n$. 



The main result of this section is the following 
statement,  turning the local
condition \eqref{suff3} on  crystals of length  at most  $m$ into  
crystallization of infinitely large crystals.

\begin{theorem}[Crystallization]\label{thm:rat} Condition
  \eqref{suff3} implies crystallization.
\end{theorem}

\begin{proof}
 Let us check that \eqref{suff}  holds  for any given $n\in \Nz$.
 Assume by contradiction
  that  this is not the case, namely that there exist  $n', \,h, \, \ell \in  \Nz$ and $ i ,\, j \in \{ 0, \dots,
  m-1\}$ with $(h+\ell)m +i + j
  =n'\leq n$ such that
\begin{equation}\label{p1}
V_{hm+i}^* + V_{\ell m +j}^* + 1 \leq V_{n'}^*.
\end{equation}
 We first observe that
\begin{align*}
m V_{n'}^*   \stackrel{\eqref{eq:basic}}{\le} V^*_{mn'} = V_{m(h+\ell)m +m i + m j}^* \stackrel{\eqref{eq:exact}}{=} m V_{hm}^*+ mV_{\ell m }^* + i V_m^* + jV_m^*. 
\end{align*}
For the following, it is convenient to set $V_0^* = 0$.  The previous inequality along with \eqref{p1} (multiplied by $m$) then yields
\begin{align*}
m \big(V_{hm}^* +  V_i^* +  V_{\ell m}^* + V_j^* +1 \big)&  \stackrel{\eqref{eq:basic}}{\le}   m \big( V_{hm+i}^* + V_{\ell m +j}^* + 1 \big) \le m V_{n'}^* \\
& \le m V_{hm}^*+ mV_{\ell m }^* + i V_m^* + jV_m^*
\end{align*}
and therefore 
\begin{equation}\label{p1-new}
V_i^* +  V_j^* + 1 \le (i+j) V_m^*/m. 
\end{equation}
Note that, if $i$ or $j$ equal zero, they can be replaced by $1$ and \eqref{p1-new} still holds as $V_0^* = V_1^* = 0$. Then, \eqref{p1-new} contradicts \eqref{suff3}.  One hence concludes that the key condition
\eqref{suff} holds  for every $n \in \Nz$,  and the assertion follows from Proposition
\ref{thm:key}. 
\end{proof}

%
%


  We close this section by presenting a hierarchy of sufficient conditions entailing
\eqref{suff3}. Let us start by considering the condition
\begin{equation}\label{suff3-old}
V^*_{sm} < V^*_p + V^*_q +  V^*_r + 1\quad \forall \, p,  q , r\ \
\text{with}  \ \ p+q+r =sm, \ s=1,2
\end{equation}
which means that the splitting of an $m$-crystal or an $2m$-crystal into three splitting parts is energetically not convenient. We check that \eqref{suff3-old} implies \eqref{suff3}. To this end, we preliminarily  note that
\begin{align}\label{eq: smaller is better}
V_r^*/r \le V_m^*/m \quad \forall \, 1 \le r \le m.\ \
\end{align}
In fact, there holds
$$r V_m^* \stackrel{\eqref{eq:exact}}{=} V_{rm}^* = V_{rm}(x_{rm}^*) = \sum_{i=0}^{m-1} V_r(x_{rm}^* + ir\alpha),$$
and therefore
$$V_r^* \le \min_{i \in \lbrace 0, \ldots, m-1 \rbrace} V_r(x_{rm}^* + ir\alpha) \le m^{-1} \sum_{i=0}^{m-1} V_r(x_{rm}^* + ir\alpha)= rV_m^*/m, $$
which implies \eqref{eq: smaller is better}. Suppose now that $1 \le p,q \le m-1$ are given. Choose $1 \le r \le m$ such that $p+q+r = sm$, $s\in \lbrace 1,2\rbrace$. Then, we get 
\begin{align*}
V^*_p + V^*_q  + 1 \stackrel{\eqref{suff3-old}}{>} V^*_{sm}- V^*_r \stackrel{\eqref{eq: smaller is better}}{\ge}  V^*_{sm} -    rV_m^*/m 
\stackrel{\eqref{eq:exact}}{=}  s V^*_{m} -    rV_m^*/m = (p+q) V_m^*/m. 
\end{align*}
This shows \eqref{suff3}.

We now consider the following stronger albeit localized 
version of condition~\eqref{suff}
\begin{equation}\label{suff4}
V^*_{u + v} < V^*_{u} + V^*_v + 1/2\quad \forall \, u,v \ \
\text{with}  \ \  u+v \leq 2m. 
\end{equation}
Condition \eqref{suff3-old}  can be deduced from the latter by subsequently
splitting the $sm$-crystal into a $(p+q)$-crystal and  an  $r$-crystal and then
splitting the $(p+q)$-crystal into a $p$-crystal and a $q$-crystal. 

 We record now different sufficient conditions entailing \eqref{suff4}. Let us start by considering a function $v_1$ with Lipschitz constant  
\begin{equation}
  \label{lip}
  \Lip v_1 < \alpha/2.
\end{equation}
Recall the definition of {\it oscillation}  of  a function
$f: \Rz \to \Rz$  as  $\osc f = \sup_{x,y} |f(x)-f(y)|$.  As
$v_1$ is nonnegative and $\min v_1=0$, we readily have that $\osc v_i = \sup v_1 $. 
 Since $|x-y|_{\rm mod \, 1}\leq 1/2$ for all $x,\, y \in
\Rz$, we find  $\osc v_1 \leq   \Lip v_1/2$.  We then get by
\eqref{lip}  
\begin{equation}
  \label{osc}
  \osc v_1 <\alpha/4.
\end{equation}
 Now, by \eqref{osc}  and $m\alpha=1$   we obtain 
  \begin{equation}
  \label{oscm}
\osc V_{p} \le  \sum_{j=0}^{p-1} \osc v_1 (\cdot + j
\alpha) \leq 2m \,\osc v_1  < 2m\alpha/4 = 1/2 \quad \forall \, p\le 2m.
\end{equation}
Let now $u+v\leq 2m$. We have that 
  \begin{align*}
    V^*_{u+v} &   \le  V_{u+v}(x^*_u)=V_u(x^*_u)+V_v(x^*_u+u\alpha) = V^*_u + V_v(x^*_u+u\alpha) \nonumber\\
&= V^*_u + V^*_v + V_v(x^*_u+u\alpha)- V_v(x^*_v)= V^*_u + V^*_v + r_{u,v}
  \end{align*}
where 
$$r_{u,v} = V_v(x^*_u+ u\alpha)- V_v(x^*_v)\leq  \osc V_v.  $$ 
In particular, \eqref{oscm} implies
\begin{equation}
  \label{r}
  r_{u,v} <1/2\ \ \forall \, u,v\ \ \text{with} \ \ u+v\leq 2m
\end{equation}
which in turn entails \eqref{suff4}. We have hence
proved the following.
\begin{proposition}[Sufficient conditions]\label{thm:hi}
\begin{equation*}
\eqref{lip} \ \Rightarrow \ \eqref{osc} \ \Rightarrow \
\eqref{oscm} \ \Rightarrow \ \eqref{r} \ \Rightarrow \
\eqref{suff4}  \ \Rightarrow \ {    \eqref{suff3-old}    \  \Rightarrow   } \ \eqref{suff3}. 
\end{equation*}
\end{proposition}

Note that all implications  in Proposition \ref{thm:hi} cannot be
reversed: 

As for \eqref{osc} $\not \Rightarrow$ \eqref{lip}, the choice
$v_1(x) = (\alpha /9)\sin (3\pi x)$ gives $\osc v_1 =
2\alpha/9<\alpha/4$ but $\Lip v_1 = 3\pi \alpha/9 >\alpha /2$. 

As for \eqref{oscm} $\not \Rightarrow$ \eqref{osc}, one can consider 
$ v_1(x) = h(1 - xm)^+ $  for $\alpha/4 \le h < 1/4$. Then  $\osc v_1
= h$  and  $\osc V_{p} \le  2h$ for all $p \le 2m$.    

As for \eqref{r} $\not \Rightarrow$ \eqref{oscm}, one takes  the sawtooth function
\begin{equation*}
v_1(x)=  \frac{2h}{\alpha}  \max_{j=0,\dots,m-1}\left(\frac{\alpha}{2}- \left|x+ j\alpha -\frac{\alpha}{2}\right| \right)
\end{equation*}
for  $h \ge 1/(4m)$  and check that $r_{u,v} = 0$ whereas $\osc V_{2m} = 2mh
\ge   1/2  $.

As for \eqref{suff4} $\not \Rightarrow$ \eqref{r},  take $\alpha = 1/2$ (i.e., $m=2$),  let
$v_1$ be locally minimized at $0, \, 1/4$, and $1/2$ with
$v_1(0)=v_1(1/2)=\varepsilon <  1/6  $, $v_1(1/4)=0$, and $v_1=h>1/2$ out
of a very
small neighborhood of $\{0, \, 1/4,\,1/2\}$.  Note that $V_1^* = 0$ and $V_p^* = \varepsilon p$ for $p \ge 2$.    Then,  \eqref{suff4}
can be readily  checked  as $V^*_{u+v}\leq (u+v)\varepsilon \leq V^*_u +
V^*_v +  2\varepsilon  < V^*_u +
V^*_v +1/2$. On the other hand, $V^*_1=0$ and $x_1^*=1/4$, so that $r_{1,1} =
V_1(x_1^*+\alpha) - V_1(x_1^*) = h - 0>1/2$.

As for  \eqref{suff3-old} $\not \Rightarrow$ \eqref{suff4}, we  take  $m=2$,  and  take 
$v_1(0)=0$ and $v_1=h$ out of a very small
neighborhood of $\lbrace 0 \rbrace$, for any $h \in (1/2,1)$. One can
easily check that $V^*_i = h\lfloor i /2 \rfloor$ for $i \in
\Nz$. Then, condition \eqref{suff3-old} holds since $V_4^* = 2h  <
h+1= V_2^* +V_1^* + V_1^*  + 1$. On the other hand, one has $V_2^* = h
> 1/2 =  V_1^* + V_1^* + 1/2$.

 As for   \eqref{suff3} $\not \Rightarrow$ \eqref{suff3-old}, we let $\alpha = 1/3$ (i.e., $m=3$), and  take 
$v_1(0)=0$ and $v_1=h$ out of a very small
neighborhood of $\lbrace 0 \rbrace$, for any $h \in (1/2,3/4)$. One
can readily compute that $V_p^*=  h (p - \lceil p/3\rceil)$ for
all $p\in\Nz$.  Thus, we can check that for $1 \le p,q \le 2$
$$(p+q) \frac{V_3^*}{3}=  h( (p-1)+ (q-1))+ h \Big(2 - \frac{p+q}{3}\Big) \le V_p^* + V_q^*  + \frac{4h}{3} <  V_p^* + V_q^* + 1.$$
However, there holds $  V_3^*  = 2h > 1 = V_1^* + V_1^* + V_1^*  +1$.

 Let us conclude this discussion by remarking that condition 
\eqref{suff3} is  indeed  not necessary for
crystallization. More precisely, let us show that \eqref{suff}   $\not
\Rightarrow$ \eqref{suff3}. Consider the previous example for $h \in
(3/4,1)$.  Recall  that $V_p^* = h(p - \lceil p/3
\rceil)$ for all $p \in \Nz$. Therefore, \eqref{suff} can be 
checked  for all $n \in \Nz$ by using  the fact that $
\lceil x/3 \rceil +  \lceil y/3 \rceil \le  \lceil (x+y)/3 \rceil+1$
for all $x, y \in \Nz$ and $h < 1$. On the other hand, \eqref{suff3}
is not satisfied since $(1+1) V_m^*/m  =  4h/3  >  1 = V_1^* + V_1^*  +1$.

\section{Crystallization for irrational $\alpha$}\label{sec:irrational}
Assume now $\alpha$ to be irrational. 
Let us start by presenting 
a necessary condition  for crystallization. 

\begin{proposition}[Necessary condition for crystallization]\label{prop:nec}
 If $ \int_0^1 v_1(t)\, \d t \geq 1$, one has no $n$-crystallization for
 $n$ large enough.
\end{proposition}

\begin{proof}
The map $x \in \Tz:=\Rz/\Zz\mapsto \alpha x\in \Tz$ is ergodic. Hence, for all $x \in  [0,1)$ we have that 
\begin{equation}
\frac{1}{n}V_n(x) =\frac{1}{n} \sum_{j=0}^{n-1}v_1(x+j\alpha)\to
\int_0^1 v_1(t)\, \d t.\label{ergo}
\end{equation}
In case $ \int_0^1 v_1(t)\, \d t \geq 1$, one has that $V_n^* > n-1$ for $n$ large enough. The statement follows
because splitting an $n$-crystal into $n$ isolated particles lowers the energy.
\end{proof}

This already shows the fundamental distinction
between the rational and the irrational. An illustration of this
difference can be obtained by fixing $m\in \Nz$ and assuming
to be given $v_1$ such that $v_1(i/m)=0$ for $i=0,\dots,m-1$ and
$\int_0^1v_1(t)\, \d t\geq 1$. Then, for all $\alpha$ such
that $m\alpha \in \Nz$ we  obtain  crystallization.  On the
other hand, if  $\alpha$ is irrational, 
there is no $n$-crystallization for $n$  large enough. 

In case the landscape potential $v_1$ exhibits some quantitative
convergence rate in \eqref{ergo}, one can deduce 
$n$-crystallization under  the  condition 
\begin{equation}
   \left| \frac1rV_r^* - \int_0^1 v_1(t)\, \d t\right| <
  \frac{1}{3r} \quad \forall \,  r \le n.  \label{ergoq}
\end{equation}
 Indeed, we have  the following result.

\begin{proposition}[$n$-crystallization]\label{eq: irrat-thm} Under \eqref{ergoq},
  $n$-crystallization holds.
\end{proposition}

\begin{proof}
For all $r\le n$, define the {\it error}
$$  e(r):= \left| V_r^* - r\int_0^1 v_1(t)\, \d t\right|.$$
Condition \eqref{ergoq} entails that $e(r)<1/3$. For all $p,\,q
\leq n $ we can
hence compute that 
  \begin{align*}
  &  V_{p+q}^*- V^*_p-V^*_q \leq (p+q)\int_0^1 v_1(t)\, \d t - p\int_0^1 v_1(t)\, \d t -
  q\int_0^1 v_1(t)\, \d t \\
&+e(p+q) + e(p)+e(q) <1  \quad \forall \,
  p, q \ \ \text{with} \ \ 
  p+q\leq n.  
   \end{align*}
In particular, the key condition \eqref{suff} holds and
$n$-crystallization follows.
\end{proof}

Note that the positive statement of Proposition \ref{eq:
  irrat-thm} is compatible with the negative assertion of
Proposition \ref{prop:nec}: by assuming \eqref{ergoq} for some $n$ and choosing
$r=1\leq n$, one has that $\int_0^1 v_1(t)\, \d
t<1/3$. 
The following  proposition  yields a sufficient condition for \eqref{ergoq}.

\begin{proposition}[Control via oscillation]\label{thm:4.3}
For all $n \in \Nz$, there holds
$$\left| \frac1nV_n^* - \int_0^1 v_1(t)\, \d t\right| \le
  \frac{\osc V_n}{n}.  $$
\end{proposition}

\begin{proof}
Since $v_1$ is $1$-periodic, we obtain
\begin{align*}
\int_0^1 v_1(t) \, {\rm d}t = \frac{1}{n}\sum_{j=0}^{n-1} \int_0^1 v_1(t + j\alpha) \, {\rm d}t = \frac{1}{n}\int_0^1 V_n(x) \, {\rm d}x.
\end{align*}
We also observe that 
$$\left| \int_0^1 V_n(x) \, {\rm d}x - V_n^*\right| \le \osc V_n. $$
The result follows by combining the two estimates.
\end{proof}

By combining Propositions \ref{eq: irrat-thm} and \ref{thm:4.3}
one gets $n$-crystallization if  $\osc V_r < 1/3$ holds for all
$r \le n$. In view of computation \eqref{oscm}, this is  in
particular   satisfied if $\osc v_1 < 1/(3n)$. 
The fundamental
difference with respect to the result in  Proposition 
\ref{thm:hi} consists in the fact that the bound on the oscillation
depends on the number of particles $n$,  and is violated for all $n$ large enough.  In the following, we seek for
conditions entailing $n$-crystallization  \emph{for all} $n$. We present a result in this direction by focusing on a special
family of piecewise constant functions  $v_1$  which are constant on
intervals of very specific length.

\begin{theorem}[Crystallization for a special piecewise constant $v_1$]\label{thm: positive}
Let $\alpha \in (0,1)$ be  irrational. Let $h >0$. For each $\varepsilon \in (0,1)$, there exists $\gamma \in (0,\varepsilon]$  such that for each open interval $I \subset (0,1)$ with $|I|=\gamma$  the $1$-periodic function $v_1$ defined by  $v_1(t) = h\chi_I(t)$ for $t \in [0,1)$ satisfies
$$ \left| \frac1nV_n(x) - \int_0^1 v_1(t)\, \d t\right| { \le }
  \frac{C_\varepsilon h}{n} \quad \forall x \in [0,1) \quad \forall \, n \in \Nz, $$
  where $C_\varepsilon$ depends on $\alpha$ and $\varepsilon$, but is
  independent of $h$ and $n$.  In particular,   by Proposition \emph{\ref{eq: irrat-thm}} this implies that    crystallization holds  if  $h < 1/(3C_\varepsilon)$. 
  \end{theorem}

It is  a standard matter to check that the above assertion holds
for  functions  $v_1$ resulting from linearly combining indicator
functions  of the type described in the statement of Theorem
\ref{thm: positive}.  More precisely, crystallization holds for
landscape potentials of the form 
$$v_1(x) = \sum_{j=1}^N h_j \chi_{ (0,\gamma_j)+\Nz}(x+x_j)$$
for any $x_1, \dots, x_N \in \Rz$ and any $h_1, \dots, h_N \in \Rz$
with  $\sum_{j=1}^N C_{\varepsilon_j}|h_j| < 1/3$,  where $\gamma_j \in
(0,\varepsilon_j]$ and $C_{\varepsilon_j}$ are given from Theorem
\ref{thm: positive},  for some $\varepsilon_1,\dots, \varepsilon_N\in (0,1)$. 

In order to prove 
Theorem \ref{thm: positive} we  need a technical lemma. 
Given an interval $I \subset (0,1)$, the {\it discrepancy} of the
sequence $\{(j\alpha)_{\mo}\}_{j\in \Nz}$ with respect to the
interval $I$ is
defined as  
\begin{align}\label{eq:phi}
\phi_n(I) =  \frac{1}{n}\sum_{j=0}^{n-1} \chi_{\Nz + I}(j\alpha) - |I|. 
\end{align}
 In the following, if not specified, $I$ may be open, half-open, or closed.

\begin{lemma}[Discrepancy control]\label{lemma: good interval length}
Let $\alpha \in (0,1)$ be irrational. For each $\varepsilon \in (0,1)$, there exists $\gamma \in (0,\varepsilon]$  such that each interval $I \subset \Rz$ with $|I|=\gamma$ satisfies 
\begin{equation}
  |  \phi_n(I)   |  \le \frac{C_\varepsilon}{n} \quad \forall \, n\in
\Nz, \label{17star}
\end{equation}
where $C_\varepsilon$ depends on $\alpha$ and $\varepsilon$. 
\end{lemma}

\begin{proof}
Fix $\varepsilon \in (0,1)$  and let $m\in \Nz$ be the smallest
integer  such that
\begin{align}\label{eq: I1}
\gamma: = m\alpha \mo \in (0,\varepsilon).  
 \end{align}
 Note that such $m$ exists uniquely since $\alpha$ is irrational. 
Consider the interval $I=x_0 + [0,\gamma)$ for some $x_0 \in \Rz$. By
choosing the constant $C_\varepsilon$ in the statement sufficiently
large,  we limit ourselves in proving \eqref{17star} for $n\ge m \lceil 1/\gamma \rceil$. Fix $k \in \Nz$ such that 
\begin{align}\label{eq: I5}
 m \lceil k/\gamma \rceil \le n  < m \lceil (k+1)/\gamma \rceil.
 \end{align}
Define  $J_n =\bigcup_{j=0}^{n-1} \lbrace j\alpha  \rbrace$ and the sets  
$$J_n^i :=   \bigcup_{\ell = 0}^{\lceil k/\gamma \rceil-1}  \lbrace  i\alpha + \ell \gamma \rbrace, \ \ \ \forall \, i=0,\ldots,m-1.$$
In view of \eqref{eq: I1} and $n \ge m \lceil k/\gamma \rceil$, we obtain 
\begin{align}\label{eq: I2}
 (J_n^i)_{\mo} \subset (J_n)_{ \mo}, \ \ \ \ \forall \, i=0,\ldots,m-1.
\end{align}
In a similar fashion, $n  < m \lceil (k+1)/\gamma \rceil$ implies
\begin{align}\label{eq: I4}
\# \left(J_n \setminus \bigcup_{i=0}^{m-1} J_n^i\right) \le m \lceil (k+1)/\gamma \rceil -  m \lceil k/\gamma \rceil \le  m(1+1/\gamma).
\end{align}
As $I = x_0+ [0,\gamma)$, it is not hard to see that
\begin{align}\label{eq: I3}
 k = \# \big(  (\Nz + I) \cap J_n^i           \big)  
\end{align}
since  $k$   consecutive intervals  in $\Nz + I$ contain exactly one element of $J_n^i$.
By \eqref{eq: I2} we get 
\begin{align*}
\sum_{i=0}^{m-1}  \# \big(  (\Nz + I) \cap J_n^i           \big) \le \sum_{j=0}^{n-1} \chi_{\Nz + I} (j\alpha )  \le   \sum_{i=0}^{m-1}  \# \big(  (\Nz + I) \cap J_n^i           \big) +  \# \left(J_n \setminus \bigcup_{i=0}^{m-1} J_n^i\right).
\end{align*}
This, along with  \eqref{eq: I4} and  \eqref{eq: I3}, shows
\begin{align*}
\left|\sum_{j=0}^{n-1} \chi_{\Nz + I} (j\alpha) - km  \right| \le  \# \left(J_n \setminus \bigcup_{i=0}^{m-1} J_n^i\right)\le m(1+1/\gamma).
\end{align*}
Since $|km/n - |I|| = |km/n - \gamma| \le m(1+\gamma)/n$ by
\eqref{eq: I5},  we estimate 
\begin{align*}
  |  \phi_n(I)   |  &\leq \left| \frac1n
  \sum_{j=0}^{n-1}\chi_{\Nz+I}(j\alpha) - \frac{km}{n}\right| +
  \left|\frac{km}{n}- |I|  \right|\\
&\leq m(1+1/\gamma)/n + m(1+\gamma)/n\leq 2m(1+1/\gamma)/n
\end{align*}
so that the statement follows for $C_\varepsilon =
2m(1+1/\gamma)$. We point
out that the  proof   can be easily adapted  for open or closed intervals of length $\gamma$ since the endpoints of the intervals appear at most once in the sequence $(j\alpha)_{j\in\Nz}$.    
\end{proof}

\begin{proof}[Proof of Theorem \emph{\ref{thm: positive}}]
Fix $\varepsilon>0$ and choose $\gamma \in (0,\varepsilon]$ as in
Lemma \ref{lemma: good interval length}. Define $I=x_0 + (0,\gamma)$
 with  $x_0 \in (0,1-\gamma)$ and $v_1(t) = h\chi_I(t)$ for $t \in [0,1)$.  
For $x \in [0,1)$, by applying Lemma \ref{lemma: good interval length}
 to  the interval $I-x$ and recalling \eqref{eq:phi}, we compute 
\begin{align*}
\left|\frac1nV_n(x) - h |I|\right| = \left|\frac{1}{n} \sum_{j=0}^{n-1}v_1(x+j\alpha) -h |I| \right| =   h\left|\frac{1}{n}\sum_{j=0}^{n-1} \chi_{\Nz + I-x} (j\alpha) -|I|\right| \le \frac{C_\varepsilon h}{n}.
\end{align*}
Since $\int_0^1 v_1(t) \, \d t = h |I|$, the statement follows.  
\end{proof}

\section{Generic noncrystallization: proof of Theorem \ref{thm: negative0}}\label{sec:gg}

This section is devoted to the proof of  Theorem
\ref{thm: negative0}. In fact,   we prove a more precise version of the
statement under the assumption that $\alpha$ is algebraic. This in turn entails Theorem
\ref{thm: negative0} by recalling that algebraic numbers are dense in the reals. We have the following.

\begin{theorem}[Generic noncrystallization]\label{thm: negative}
Let $\alpha \in (0,1)$ be irrational and algebraic. For each $\varepsilon>0$ and  each $1$-periodic,  piecewise continuous, and  lower semicontinuous function $v_1$ with $\min v_1 = 0$, there exists a  $1$-periodic, piecewise constant, and  lower semicontinuous function $v_1^\varepsilon$  with $\min v_1^\varepsilon = 0$  
such that  $\Vert v_1 - v_1^\varepsilon \Vert_{L^\infty(0,1)} \le
\varepsilon$  and  a strictly increasing sequence $(n_k)_{k\in
  \Nz}  \in 2 \Nz$ satisfying
$$
(V^\varepsilon_{n_k})^* > (V^\varepsilon_{n_k/2})^* + (V^\varepsilon_{n_k/2})^* +1,
$$
 where $ (V^\varepsilon_{n_k})^* := \min \sum_{j=0}^{n_k-1}
 v_1^\varepsilon(x+j\alpha)$. Consequently, it is energetically
 favorable to split an $n_k$-crystal into two $n_k/2$-crystals 
 and no $n_k$-crystallization holds. 
\end{theorem}

 The proof of Theorem \ref{thm: negative} is split in a series of
lemmas. The statement  crucially  relies on  some properties
of the discrepancy of additive recurrent sequences,  recall the
definition \eqref{eq:phi}. Firstly,
 for $\alpha$ irrational, \cite[Thm.\ 1.51]{Drmota} yields   a infinite subset $\mathcal{N}_\alpha \subset \Nz$  such that 
\begin{align}\label{eq: fund-lowerbound}
\sup_{x \in(0,1)}  |\phi_{n}([0,x))| \ge  C_0 \frac{\log n}{n} \quad \forall \, n\in \mathcal{N}_\alpha
\end{align}
for a universal constant $C_0>0$. 
  Secondly,  if $\alpha$ is also  algebraic, for each $\eta>0$ there exists $C_\eta>0$ such that the upper bound  
 \begin{align}\label{eq: fund-upperbound}
\sup_{0<x<y<1}  |\phi_n((x,y))| \le C_\eta n^{-1 + \eta} \quad \forall \, n\in \Nz
\end{align}
holds  \cite[Thm.\ 3.2 and Ex.\ 3.1, pp. 123-124]{Kuipers}. 
 
 Our first task is that of showing that, by possibly changing the
 constant, a lower bound like
 \eqref{eq: fund-lowerbound} holds not only  at  one  single  specific  point,  but  that each
 interval of arbitrarily small length contains  at least one 
 point fulfilling \eqref{eq: fund-lowerbound}.

\begin{lemma}[Lower bound]\label{lemma: Neps}
 Let $\alpha$ be irrational and let $\mathcal{N}_\alpha \subset \Nz$ as in \eqref{eq: fund-lowerbound}.     Let $\varepsilon>0$. Then there exists $N_\varepsilon \in \Nz$ such that for  all $n \in \mathcal{N}_\alpha$ with $n \ge N_\varepsilon$  each interval $I\subset (0,1)$ with  $|I| \ge \varepsilon$ contains a point $x \in I$  with
\begin{equation}
 |\phi_n([0,x))| \ge C_0 \frac{\log n}{2n}.\label{eq:
  fund-lowerbound0} 
\end{equation}
\end{lemma}

\begin{proof}
Given $\varepsilon>0$, choose $\gamma \in (0,\varepsilon]$ as in Lemma \ref{lemma: good interval length}. Select $N_\varepsilon$ sufficiently large such that 
\begin{align}\label{eq: Neps}
\log N_\varepsilon \ge \frac{6C_\varepsilon}{C_0\gamma},
\end{align}
where $C_\varepsilon$ is the constant of Lemma \ref{lemma: good interval length}.  Let $n \in \mathcal{N}_\alpha$ with   $n \ge N_\varepsilon$.  In view of \eqref{eq: fund-lowerbound},  we choose $x_0 \in (0,1)$ satisfying
\begin{align}\label{eq: specific choice}
|\phi_n([0,x_0))| \ge C_0 \frac{2\log n}{3n}.
\end{align}
Consider the collection of points $x_k  =   x_0  + k \gamma $, $k \in \Zz$,  with $x_k \in (0,1)$.  Note that $x_k \in (0,1)$ and $x_0 \in (0,1)$ imply   $|k| \le 1/\gamma$. For each $x_k  \in (0,1)$, $k \ge 1$, we observe that 
\begin{align*}
\phi_n([0,x_k)) = \phi_n([0,x_0)) + \sum_{l=1}^{k} \phi_n([x_{l-1},x_{l})).
\end{align*}
Therefore, since $|\phi_n([x_{l-1},x_{l}))| \le C_\varepsilon/n$ for
$l=1,\ldots,k$ by  Lemma \ref{lemma: good interval length}, we derive
 from estimates  \eqref{eq: Neps}, \eqref{eq: specific
  choice}, and $k \le 1/\gamma$  that 
$$|\phi_n([0,x_k))| \ge{  |  \phi_n([0,x_0))  | } -  \sum_{l=1}^{k} { |\phi_n([x_{l-1},x_{l}))| }  \ge C_0 \frac{2\log n}{3n} -  \frac{C_\varepsilon}{n\gamma}  \ge C_0 \frac{\log n}{2n}.$$
The same estimate holds for  each $x_k \in (0,1)$,   $k \le -1$.  The result now follows from
the fact that each interval in $(0,1)$ of length at least
$\varepsilon$ contains at least one of the points  $x_k$.
\end{proof}

 Next, we  show that,  by possibly reducing the
constant, a point fulfilling a lower bound like \eqref{eq: fund-lowerbound0}   can be chosen independently of $n$.

\begin{lemma}[Lower bound, independent of $n$]\label{lemma: N sequence}
Let $\delta \in (0,1)$. There exists a strictly increasing sequence of
integers $(n_k)_{k\in \Nz}  \subset \Nz$ and $x \in (0,\delta)$ such that
$$ |\phi_{n_k}([0,x))| \ge C_0 \frac{\log n_k}{4n_k} \quad \forall \, k \in \Nz.$$
\end{lemma}

\begin{proof}
We define the sequence of integers  $(n_k)_{k\in \Nz}$ iteratively. As first step of the iteration procedure, apply Lemma \ref{lemma: Neps} to
$(0,\delta)$ with $\varepsilon =\delta$. This gives $x_1 \in
(0,\delta)$ and $N_\varepsilon $ such that, by letting
$n_1=  \min \lbrace n \in \mathcal{N}_\alpha\colon n \ge N_\varepsilon\rbrace  $ one has
$|\phi_{n_1}([0,x_1))| \geq C_0 {\log n_1}/{(2n_1)}$. In case
$\phi_{n_1}([0,x_1))>0$, define $I_{1} := (x_1,
x_1 + \delta_1)$ for some $0<  \delta_{1} \le C_0{\log
  n_1}/{(4n_1)}$  so small  that $I_{1} \subset
(0,\delta)$. We can then compute for all  $x \in I_{1}$
\begin{align}
\phi_{n_{1}}([0,x)) & = \frac{1}{n}\sum_{j=0}^{n-1} \chi_{\Nz + [0,x)}(j\alpha) - |[0,x)| \ge \frac{1}{n}\sum_{j=0}^{n-1} \chi_{\Nz + [0,x_1)}(j\alpha) - |[0,x_1)|  - \delta_1 \nonumber\\
&= \phi_{n_1}([0,x_1)) - \delta_1 \ge   C_0 \frac{\log n_{1 }}{2n_{1}} - \delta_1   \ge    C_0 \frac{\log  n_1  }{4n_1}.\label{eq:arg}
\end{align}
 If
$\phi_{n_1}([0,x_1))<0$ instead, we repeat the argument for
$I_{1} := (x_1 - \delta_{1} , x_1) \subset (0,\delta)$
in place of $I_{1} = (x_1, x_1 + \delta_{1})$.

 Suppose now  that for $\ell \in \Nz$ there
exists  a strictly  increasing set of integers $n_k $, $1 \le k \le
\ell-1$, and nested intervals $I_{\ell-1} \subset I_{\ell-2} \subset
\ldots I_1 \subset  (0,\delta)$   such that
\begin{align}\label{eq: induction}
 |\phi_{n_k}([0,x))| \ge C_0 \frac{\log n_k}{4n_k} \quad \forall \, x \in I_{\ell-1} \quad \forall \, 1 \le k \le \ell-1.
  \end{align}
We have already checked above that \eqref{eq: induction} can be
realized for $\ell
=2$. 

 We now define $n_\ell$ and $I_\ell$ as follows.  Fix $\varepsilon
  \leq 
 |I_{\ell-1}|$. 
 By applying Lemma \ref{lemma: Neps} to interval $ I_{\ell-1}$
with $\varepsilon$ one
finds $x_\ell \in I_{\ell-1}$ and $N_\varepsilon$ such that, by letting  $n_{\ell} =
  \min \lbrace n \in \mathcal{N}_\alpha\colon n \ge  N_\varepsilon \text{ and } n \ge n_{\ell-1} +1 \rbrace$ one has 
\begin{align}\label{eq:: xell}
|\phi_{n_{\ell}}([0,x_\ell))| \ge C_0 \frac{\log n_{\ell }}{2n_{\ell}}.
\end{align} 

We now construct $I_\ell$  by arguing as above: if 
$\phi_{n_{\ell}}([0,x_\ell))>0$,  we  define $I_{\ell} := (x_\ell,
x_\ell + \delta_{\ell})$ for some $0<  \delta_{\ell} \le C_0{\log
  n_\ell}/{(4n_\ell)}$  so small  that $I_{\ell} \subset
I_{\ell-1}$.  By arguing as in \eqref{eq:arg} with the help of  \eqref{eq:: xell}, we  then compute   for all $x \in I_{\ell}$
\begin{align*}
\phi_{n_{\ell}}([0,x)) \ge    C_0 \frac{\log n_\ell}{4n_\ell}.
\end{align*}
 Along   with the induction hypothesis \eqref{eq: induction}
for $\ell -1$,  this   shows that \eqref{eq: induction} holds for all $x \in I_\ell$ and all $1 \le k \le \ell$. If $\phi_{n_{\ell}}([0,x_\ell))<0$ instead, we repeat the argument for $I_{\ell} := (x_\ell - \delta_{\ell} , x_\ell) \subset I_{\ell -1}$  in place of $I_{\ell} = (x_\ell, x_\ell + \delta_{\ell})$. 

By performing this construction for each $k \in \Nz$,  we obtain a
sequence $(x_k)_{k\in \Nz}$ and nested intervals $(I_k)_{k\in \Nz} \subset (0,1)$. Since $|I_k| \to 0$, we
 have   that $x_k \to x$,  where $x$ is the point with $\lbrace x \rbrace =  \bigcap_{k=1}^\infty I_k$.  The statement now follows from \eqref{eq: induction} and the fact that $x \in I_\ell$ for all $\ell \in \Nz$.    
\end{proof}

Next, we construct an approximation of $v_1$ such that the sufficient condition for $n$-crystallization \eqref{ergoq} is violated.

\begin{lemma}[Approximation of $v_1$]\label{th: bad approx}
Let $\alpha \in (0,1)$ be irrational. There exists a strictly
increasing sequence of integers $(n_k)_{k\in \Nz}$ such that the following
holds: for each $\varepsilon>0$ and  each $1$-periodic,  piecewise continuous, and  lower semicontinuous function $v_1$ with $\min v_1 = 0$, there exists a  $1$-periodic, piecewise constant, and  lower semicontinuous function $v_1^\varepsilon$  with $\min v_1^\varepsilon = 0$ 
such that  $\Vert v_1 - v_1^\varepsilon \Vert_{L^\infty(0,1)} \le \varepsilon$  and
\begin{align}\label{eq: bad approx0}
 \frac{1}{n_k}(V^\varepsilon_{n_k})^* - \int_0^1  v_1^\varepsilon(t) \, {\rm d}t \le - C \frac{\log n_k}{n_k}  \quad \forall k \in \Nz,
\end{align}
for some $C>0$ only depending on $\alpha$, $\varepsilon$, and $v_1$, where $ (V^\varepsilon_{n_k})^* := \min \sum_{j=0}^{n_k-1} v_1^\varepsilon(x+j\alpha)$.
   \end{lemma}

\begin{proof}
Our goal is approximate $v_1$ by a piecewise constant function with the desired property. Fix $\delta >0$ small.  We apply Lemma \ref{lemma: N sequence} to get a sequence $(n_k)_{k\in \Nz}$ and $x_1 \in (0,\delta)$ such that
\begin{align}\label{eq: bad approx1}
|\phi_{n_k}([0,x_1))| \ge \frac{C_0 \log n_k}{4n_k}\quad \forall\,  k \in \Nz
\end{align}
holds. Moreover, we choose $\gamma \in (0, \delta]$ as in Lemma
\ref{lemma: good interval length} (applied for $\delta$ in place of
$\varepsilon$), and we decompose the interval $(0,1)$ by means of
 the points  $x_0 = 0 < x_1 < x_2 < \ldots < x_l =1$ such that 
\begin{align*}
x_i - x_{i-1} = \gamma  \quad \forall \, i = 3, \ldots,l.
\end{align*}   
Note that $l \le 2 + 1/\gamma$  and that $|x_i - x_{i-1}| \le \delta$ for $i=1,\ldots,l$.  Since Lemma  \ref{lemma: good interval length} implies $|\phi_{n_k}([x_{i-1},x_i))| \le C_\delta/n_k$ for all $i=3,\ldots,l$, we find 
\begin{align}\label{eq: bad approx3}
\big|\phi_{n_k}([0,x_1)) + \phi_{n_k}([x_1,x_2)) \big| \le  \left |\sum_{i=3}^l  \phi_{n_k}([x_{i-1},x_i)) \right|  \le \frac{C_\delta}{\gamma n_k}.    
\end{align}
Without restriction we treat the case $\phi_{n_k}([0,x_1))>0$. The
other case is similar  but requires  a different notational realization. 
We define a  piecewise constant, $1$-periodic function $v_1^\varepsilon$ by setting
\begin{align}\label{eq: bad approx4}
v_1^\varepsilon = b_i \quad \text{ on } (x_{i-1},x_i) \quad \text{ for
  } \ i =1 \ldots,l 
\end{align}
for suitable $b_i \in [0,M]$, where $M:= \sup v_1$. The values at
$x_i$, $i=0\ldots,l-1$, can be chosen in such a way that the function is  lower semicontinuous. Recall $|x_i - x_{i-1}| \le \delta$ for $i=1,\ldots,l$. Thus,  given $\varepsilon >0$, we observe that by choosing $\delta=\delta(\varepsilon)>0$ sufficiently small and the values $b_i$ appropriately, we can achieve   $\Vert v_1^\varepsilon - v_1 \Vert_{L^\infty(0,1)} \le \varepsilon$.  Moreover, this can be done in such a way that $b_1 < b_2$  and that $\min_i b_i = 0$.   We now check \eqref{eq: bad approx0}. First, in view of definition \eqref{eq: bad approx4} and \eqref{eq:phi}, we compute
\begin{align}\label{eq: bad approx5}
 &\min_{x\in  [0,1)  } \frac{1}{n_k}\sum_{j=0}^{n_k-1} v_1^\varepsilon(x+j\alpha) - \int_0^1 v_1^\varepsilon(t) \,{\rm d}t \le \frac{1}{n_k}\sum_{j=0}^{n_k-1} v_1^\varepsilon(j\alpha) - \int_0^1 v_1^\varepsilon(t) \,{\rm d}t \notag \\
 & \quad\le \sum_{i=1}^l  b_i\Big(   \frac{1}{n_k}\sum_{j=0}^{n_k-1}   \chi_{[x_{i-1},x_{i}) + \Nz} (j\alpha  )   -  |x_i-x_{i-1}|\Big)\notag
 \\ &  \quad=  \sum_{i=1}^l   b_i \, \phi_{n_k}([x_{i-1},x_i)).
\end{align}
Moreover, by $\max_i b_i \le M$, $b_1 - b_2<0$, and
$\phi_{n_k}([0,x_1))>0$   we obtain  from estimates 
\eqref{eq: bad approx1} and \eqref{eq: bad approx3}  that 
\begin{align*}
 &\sum_{i=1}^l   b_i \, \phi_{n_k}([x_{i-1},x_i)) \le   (b_1 - b_2) \phi_{n_k}([x_{0},x_1) ) +   M \big|\phi_{n_k}([x_{0},x_1)) + \phi_{n_k}([x_{1},x_2))\big| \\
 & \quad+   M   \sum_{i=3}^l \big| \phi_{n_k}([x_{i-1},x_i)) \big| \\
 & \quad\le   (b_1 - b_2) \frac{C_0 \log n_k}{4n_k}  + \frac{2M C_\delta}{\gamma n_k}.
\end{align*}
This along with \eqref{eq: bad approx5} shows 
$$ \min_{x\in [0,1)  } \frac{1}{n_k}\sum_{j=0}^{n_k-1} v_1^\varepsilon(x+j\alpha) - \int_0^1 v_1^\varepsilon(t) \,{\rm d}t \le  - C\frac{\log n_k}{n_k} $$
for all $k \in \Nz$ for some suitable $C>0$ only depending on $\alpha$, $\delta$, and $v_1$.
 \end{proof}

 Eventually, we establish  the following upper bound.

\begin{lemma}[Upper bound]\label{lemma: upper bound}
Let $\alpha \in (0,1)$ be irrational and algebraic. Let $v_1$ be a   $1$-periodic, lower semicontinuous function of the form $v_1 = \sum_{i=1}^k b_i \chi_{I_i}$ for intervals $I_i \subset (0,1)$. Then there holds 
$$
\left| \frac{1}{n}V_{n} - \int_0^1  v_1(t) \, {\rm d}t  \right|\le C_\eta \, n^{-1+ \eta} \quad \forall \, n \in \Nz,
$$
for some $C_\eta>0$ only depending on $\alpha$, $\eta$, and $v_1$. 
   \end{lemma}

\begin{proof}
 Arguing  as in the proof  of the approximation Lemma \ref{th: bad approx}, we calculate
$$\left| \frac{1}{n} V_{n}  - \int_0^1  v_1(t) \, {\rm d}t  \right| \le \big(\max_i b_i\big) \sum_{i=1}^k { | \phi_n(I_i)|,   }$$
where $\phi_n(I_i)$ is defined in \eqref{eq:phi}. The statement follows from the fact that $| \phi_n(I_i)  |  \le C_\eta n^{-1+\eta}$ for all $n \in \Nz$, see \eqref{eq: fund-upperbound}.
\end{proof}

We are finally in  the position of proving Theorem \ref{thm: negative}.

\begin{proof}[Proof of Theorem \emph{\ref{thm: negative}}.]

Given $v_1$,  $\alpha \in (0,1)$ irrational and algebraic, and
$\varepsilon>0$, we define the   piecewise constant,  lower semicontinuous function
$v_1^\varepsilon$ as in   the approximation Lemma \ \ref{th: bad
  approx}. We  aim at showing  that for each $n' \in \Nz$ we
 can  find $n \in 2\Nz$ with $n  \ge n'$ such that it is energetically favorable to split an $n$-crystal into two $n/2$-crystals, i.e.,
\begin{align}\label{eq: main claim}
(V^\varepsilon_{n})^* > (V^\varepsilon_{n/2})^* + (V^\varepsilon_{n/2})^* +1.
\end{align}
For all $r\in \Nz$, define the \emph{error} as 
\begin{align}\label{eq: error}
e(r) :=   (V^\varepsilon_r)^* - r\int_0^1  v_1^\varepsilon  (t)\, \d t.
\end{align}
Fix $n' \in \Nz$. By the approximation Lemma  \ref{th: bad approx} we find $n_0 \ge n'$ such that $e(n_0) < - C \log n_0$ for $C$ only depending on $\alpha$, $\varepsilon$, and $v_1$. We can suppose that $n_0$ is chosen large enough such that  $e(n_0) \le -2$.   We claim that there exists $k \in \Nz$ such that 
\begin{align}\label{eq: intermediate claim}
2 \, e(n_0 2^{k-1})   +1 < e(n_0 2^{k}).
\end{align}
 In fact, assume that this  was  not the case. Then,  we would
have  $2 e(n_0 2^{k-1})   +1 \ge  e(n_0 2^{k})$ for all $k \in
\Nz$. Consequently, by an  iterative application of this
estimate and  by  using $e(n_0) \le -2$  we get
\begin{align*}
e(n_0 2^k) \le 2^k e(n_0) + 2^k-1 \le - 2^k  \quad \forall \, k \in \Nz.
\end{align*}
This contradicts the fact that $|e(n_0 m)| \le C_\eta (n_0m)^{\eta}$
for all $m \in \Nz$,  as predicted by the upper bound in 
Lemma \ref{lemma: upper bound}. Thus, \eqref{eq: intermediate claim}
holds for some $k \in \Nz$. Set  now  $n:= n_0 2^k$  and
use \eqref{eq: error} to compute 
\begin{align*}
(V^\varepsilon_{n})^* =  e(n) +n\int_0^1 v^{ \varepsilon  } (t)\, \d t > 1 + 2 \left( e(n/2)  + \frac{n}{2}\int_0^1 v^{ \varepsilon }_1(t)\, \d t \right) = 1+ 2(V^\varepsilon_{n/2})^*. 
\end{align*}
This shows \eqref{eq: main claim} and concludes the proof.     
\end{proof}

\section*{Acknowledgements}  
  MF is supported by  the Deutsche Forschungsgemeinschaft (DFG, German
  Research Foundation) under Germany's Excellence Strategy EXC
  2044-390685587, Mathematics M\"unster:
  Dynamics--Geometry--Structure. US is partially supported by the 
  Austrian Science Fund (FWF) 
   project   F\,65 and by the 
Vienna Science and Technology Fund (WWTF) project MA14-009.

\bibliographystyle{alpha}

\begin{thebibliography}{99}



 
\bibitem{Assoud}
L. Assoud, R. Messina, H. L\"owen. Stable crystalline lattices in two-dimensional
binary mixtures of dipolar particles. {\it Europh. Lett.} 80 (2007),
1--6. 


\bibitem{Assoud2}
L. Assoud, R. Messina, H. L\"owen. Binary crystals in
two-dimensional two component Yukawa mixtures. {\it J. Chem. Phys.}
129 (2008), 164511.






\bibitem{Betermin}
{L.~B\'etermin, H.~Kn\"upfer, F.~Nolte}.
\newblock  {Crystallization of  one-dimensional alternating two-component systems}. 
\newblock Preprint at \href{https://arxiv.org/abs/1804.05743}{\tt arXiv:1804.05743}.






\bibitem{Blanc}
{X.~Blanc, M.~Lewin}.
\newblock {The crystallization conjecture: a review}.
\newblock {\it EMS Surv.\ Math.\ Sci.}  
\newblock {2} (2015), 255--306.


















\bibitem{Drmota} M. Drmota, R. F. Tichy. {\it Sequences, discrepancies and applications}. Lecture Notes in Mathematics, 1651, Springer, Berlin, 1997.








 
\bibitem{Eldridge}
M. D. Eldridge, P. A. Madden, D. Frenkel.
Entropy-driven formation of a superlattice in a hard-sphere binary mixture.
{\it Nature} 365 (1993), 35--37 


\bibitem{Fanzon}
S. Fanzon, M. Ponsiglione, R. Scala.
Uniform distribution of dislocations in Peierls-Nabarro models for
semi-coherent interfaces. Submitted, 2019. Preprint at  \href{https://arxiv.org/abs/1908.04222}{\tt  arXiv:1908.04222}.







\bibitem{kreutz}
{M.~Friedrich, L. Kreutz}. 
\newblock {Crystallization in the hexagonal lattice for ionic
  dimers}. {\it Math. Models Meth. Appl. Sci.} (2019), to appear. Preprint at \href{https://arxiv.org/abs/1808.10675}{\tt
  arXiv:1808.10675}.



\bibitem{kreutz2}
{M.~Friedrich, L.~Kreutz}. 
\newblock  {Finite crystallization and Wulff shape emergence for ionic compounds
    in the square lattice}.  Submitted, 2019. Preprint at  \href{https://arxiv.org/abs/1903.00331}{\tt  arXiv:1903.00331}.








\bibitem{Friesecke-Theil15}
{G.~Friesecke, F.~Theil}. 
\newblock {Molecular geometry optimization,
  models}. In the {\it Encyclopedia of Applied and Computational Mathematics},
B. Engquist (Ed.), Springer, 2015.



 \bibitem{Gardner}
C. S. Gardner, C. Radin.
The infinite-volume ground state of the Lennard-Jones potential. {\it
  J. Stat. Phys.} 20 (1979), 719--724.







\bibitem{Hamrick} G. C. Hamrick, C. Radin. The symmetry of
  ground states under perturbation. {\it J. Stat. Phys.} 21 (1979),
  601--607.








\bibitem{Heitman-Radin80}
R. Heitman, C. Radin.
Ground states for sticky disks. {\it J. Stat. Phys.} 22 (1980), 281--287.


 






\bibitem{Kuipers} L. Kuipers, H. Niederreiter. {\it Uniform
    distribution of sequences}. Pure and Applied
  Mathematics. Wiley-Interscience, John Wiley \& Sons, New York-London-Sydney, 1974.




\bibitem{Levi}
E. Levi, J. Minar, I. Lesanovsky. Crystalline structures in a one-dimensional
two- component lattice gas with $1/r\alpha$ interactions. {\it
  J. Stat. Mech. Theor. Exp.} (2016),  033111.

\bibitem{Lewars}
{E.~G.~Lewars}.
\newblock {\em Computational Chemistry}. 2nd edition, Springer, 2011.









\bibitem{Radin86} C. Radin. Crystals and quasicrystals: a continuum model. {\it Comm. Math. Phys.} 105 (1986), 385--390.












 

   





  \bibitem{Ventevogel}
{ W.~J.~Ventevogel, B.~R.~A.~Nijboer}. 
\newblock {On the configuration of systems of interacting atom with
minimum potential energy per atom}. 
\newblock  {\it Phys.\ A} 99 (1979), 565--580.










 



\bibitem{Xu}
H. Xu, M. Baus. A density functional study of superlattice formation in binary
hard-sphere mixtures. {\it  J. Phys. Condens. Matter} 4 (1992),
L663.










  
  
 







  

\end{thebibliography}

\end{document}